\documentclass[journal]{IEEEtran}
\usepackage[pdftex]{graphicx}
\usepackage{wrapfig}
\usepackage{subfigure}
\usepackage[symbol]{footmisc}
\usepackage{bm}
\usepackage{bbm}
\usepackage{amsmath}
\usepackage{cases}
\usepackage{color}
\usepackage{outlines}
\usepackage{url}

\usepackage{amssymb}

\usepackage{cite}

\ifCLASSINFOpdf
\else
\fi
\usepackage{amsmath,amssymb,amsthm}

\begin{document}

%
\title{How likely is a random graph shift-enabled?}
%
%
%

\author{Liyan~Chen,
Samuel~Cheng$^1$,~\IEEEmembership{Senior~Member,~IEEE,}
                Vladimir~Stankovic,~\IEEEmembership{Senior~Member,~IEEE,}
        and~Lina~Stankovic,~\IEEEmembership{Senior~Member,~IEEE}
\thanks{L. Chen is with 
Key Laboratory of Oceanographic Big Data Mining \& Application of Zhejiang Province, Zhejiang Ocean University, Zhoushan, Zhejiang 316022, China
and the Department
of Computer Science and Technology, Tongji University, Shanghai,
 201804 China  (e-mail: chenliyan@tongji.edu.cn).}
\thanks{S. Cheng is with
the School of Electrical and Computer Engineering, University of Oklahoma, OK 74105, USA (email: samuel.cheng@ou.edu).}
\thanks{V. Stankovic and L. Stankovic are with Department of Electronic and Electrical Engineering,
         University of Strathclyde, Glasgow, G1 1XW U.K.
        (e-mail:\{vladimir.stankovic,~lina.stankovic\}@strath.ac.uk).}
\thanks{$^1$ Corresponding author.}
\thanks{
}
}

%
%

\markboth{}%
{Shell \MakeLowercase{\textit{et al.}}: Bare Demo of IEEEtran.cls for IEEE Journals}
%



\maketitle

\begin{abstract}
The shift-enabled property of an underlying graph is essential in designing distributed filters. This article discusses when a random graph is shift-enabled. In particular, popular graph models Erdős–Rényi (ER), Watts–Strogatz (WS), Barabási–Albert (BA) random graph are used, weighted and unweighted, as well as signed graphs. Our results show that the considered unweighted connected random graphs are shift-enabled with high probability when the number of edges is moderately high. However, very dense graphs, as well as fully connected graphs, are not shift-enabled. Interestingly, this behaviour is not observed for weighted connected graphs, which are always shift-enabled unless the number of edges in the graph is very low. 

\end{abstract} %


%
\IEEEpeerreviewmaketitle

\newtheorem{myDef}{Definition}
\newtheorem{Thm}{Theorem}
\newtheorem{Rem}{Remark}
\newtheorem{Lem}{Lemma}
\newtheorem{Cor}{Corollary}

\section{Introduction}
\label{sec:intro}
Graph signal processing (GSP) extends classical digital signal processing to signals on graphs and provides a promising solution to numerous real-world problems that involve data defined on topologically complicated domains~\cite{Liyanchen2018}. For large graphs, graph signals need to be processed in a distributed rather than centralized manner \cite{sandryhaila_2014_big_data}. That is, a graph node may only have access to graph signals acquired by nodes in its physical proximity. Furthermore, for large graphs with millions of nodes, a centralised implementation of
the graph filter \cite{Shuman_2013_The_emerging_field}, \cite{ortega2018graph} through direct matrix multiplication is computationally intractable~\cite{segarra2017optimal,coutino2019advances}.
Thus to make the graph filtering feasible, it is necessary to perform
the filtering operation locally \cite{sandryhaila_2014_big_data}. 
For practical design purposes,
it is necessary to be in a position to decompose graph filters in
a form of polynomial of a shift matrix, of graph shift operator, $S$, that uniquely defines graph topology
(for example, graph adjacency or Laplacian matrix)
\cite{sandryhaila_2013_discrete}, \cite{segarra2017network}. 
However, not all graph filters can be represented as polynomials of the shift matrix\footnote{The importance of this polynomial representation has been reiterated in a recent survey paper \cite{ortega2018graph}.}.



Given a graph, necessary conditions for a graph filter to be representable as a polynomial of the graph shift matrix is discussed in \cite{Liyanchen2018} and \cite{Liyanchen2021}, where the notion of \emph{shift-enabled} graph is introduced as a graph where any shift-invariant filter $H$ can be represented as a polynomial of the shift matrix.
It is shown in \cite{Liyanchen2018} and \cite{Liyanchen2021} that the shift-enabled condition \cite{Liyanchen2018} is important for both directed and undirected graphs, and hence it needs to be taken into account.  

This paper focuses on finding a likelihood for a random graph to be shift enabled. 
This problem has received relatively little attention in the research community, since most researchers currently assume that the shift-enabled
condition simply holds or ignore the condition completely.
To illustrate ``how likely is a graph shift-enabled'',
we discuss the probability that some classic random graphs are shift-enabled. 
In particular, the main contribution of this paper is characterising the bahaviour of the probability that:
\begin{itemize}
	\item an unweighted random Erdős–Rényi (ER) graph, Watts–Strogatz (WS) graph and Barabási–Albert (BA) graph is shift-enabled as a function of graph parameters;
	\item the above three weighted random graphs are shift-enabled, where the weights are generated based on exponential and Gaussian distribution;
	\item a random signed graph is shift-enabled; 
	\item the analysis of the above results.
\end{itemize}
Our results show that the considered unweighted connected random graphs are shift-enabled with high probability when the number of edges is moderately high. However, very dense graphs, as well as fully connected graphs, are not shift-enabled. Interestingly, this behaviour is not observed for weighted connected graphs, which are always shift-enabled unless the number of edges in the graph is very low.

The outline of the paper is as follows. Section
\ref{sect:basicconcepts} describes the basic concepts and fundamental properties of a shift-enabled graph. Section~\ref{sec:method} provides the main results of the paper for unweighted graphs, that is, the characterisation of the behaviour of the probability that a graph is shift enabled. Section~\ref{sec:weights} extends the results to weighted and signed graphs.
Section~\ref{sec:discussion} concludes the paper.


\section{Basic concepts and properties of shift-enabled graphs}
\label{sect:basicconcepts}
In this section, we introduce our notation and briefly review the concepts of shift-enabled graphs and their properties relevant to this article. For more details, see \cite{sandryhaila_2014_big_data,sandryhaila_2014_discrete_frequency, Shuman_2013_The_emerging_field,sandryhaila_2013_discrete}.


Let $\mathcal{G}=(V,E,W)$ be a graph, where $V=\{v_1,v_2,\cdots,v_{N}\}$ is the vertex set, and $E\subset\{1,\cdots,N\}\times\{1,\cdots, N\}$ is the edge set in which $(i,j)\in E$ if vertex $v_i$  and vertex $v_j$ have a link. $W=(w)_{i,j=1}^N$ is the weighted adjacency matrix, in which  $w_{i,j}$ represents the weight of the edge $(i,j)\in E$.
 Throughout this article, a graph $\mathcal{G}$ is assumed to be simple, i.e. a finite, graph without loops and/or multiple edges. 
Let $D=diag(D_1,\cdots,D_n)$, with $D_i=\sum_{j=1}^{N}w_{i,j}$, be the degree matrix of $\mathcal{G}$\cite{poignard2018spectra}.


The graph is shift-enabled if its shift matrix $S$ (see Remark 2 below) complies with the following definition.


\begin{myDef}[Shift-enabled graph {\cite{Liyanchen2018,sandryhaila_2013_discrete}}]
	\label{def:shift_enable}
	A graph ${\mathcal G}$ is shift-enabled if its corresponding shift matrix $S$ satisfies $p_S(\lambda) =m_S(\lambda)$, where $p_S(\lambda)$ and $m_S(\lambda)$ are the minimum polynomial and the characteristic polynomials of $S$, respectively. We also say that $S$ is shift-enabled when the above condition is satisfied. Otherwise, $S$ and the corresponding graph, are non-shift-enabled.
\end{myDef}

\begin{Rem}
The shift-enabled condition ($p_S(\lambda) =m_S(\lambda)$) is equivalent to the condition that each Jordan block of the Jordan normal form of the shift matrix has a distinct eigenvalue (see Proposition 6.6.2 in \cite{lancaster_1985_matrix_theory}).
Consequently, for real symmetric matrices, the above condition naturally degenerates to the condition that all eigenvalues have to be distinct (see details in Lemma  ~\ref{Lem_undirected_shift}).
\end{Rem}

\begin{Rem}
The graph adjacency matrix $A$, 
Laplacian $L=D-W$,  the normalized Laplacian matrix $\mathcal{L}=L^{-1/2} D L^{-1/2}$, the signless Laplacian matrix $L^{+}=D+W$  and the probability transition matrix $T=D^{-1}W$ are generally chosen as the graph shift operator or graph shift matrix~\cite{sandryhaila_2014_big_data, sandryhaila_2014_discrete_frequency,Shuman_2013_The_emerging_field,Marques_2017}. Here,  we use $S$ to denote the general
shift matrix 
and select Laplacian matrix as the shift matrix for the specific discussion in Section~\ref{sec:method} and Section~\ref{sec:weights}, since
Laplacian matrix is one of most popular shift matrix~\cite{Marques_2017,dong2020graph,Dittrich20200signedgraph}. 
Most of the conclusions in this paper, however, apply to other shift matrices (see Figure~\ref{fig:different shift matrix}).
\end{Rem}





For shift-enabled graphs, we have the following result.
\begin{Thm}{\label{Thm_shift_commute}}
	The shift matrix $S$ is shift-enabled if and only if every matrix $H$ commuting with $S$ is a polynomial in $S$\rm~{\cite{sandryhaila_2013_discrete}}.
\end{Thm}
A graph filter $H$ is linear shift-invariant (LSI) if $H$ commutes with shift matrix $S$  ($HS=SH$).
That is, the shifted and filtered operations are commuted, i.e., the shifted filtered output is the same as filtered output of a shifted input. 
Theorem~\ref{Thm_shift_commute} implies that 
an LSI filter naively designed cannot always be represented as a polynomial of shift
operators as is the case in classical DSP; that is, as mentioned in Section~\ref{sec:intro}, whether the graph is directed or undirected, the shift-enabled condition is important. Thus, it is interesting to investigate ``how likely is a graph shift-enabled?''. The next section gives the answers for commonly used random graphs.

\section{Unweighted random graphs}
\label{sec:method}
In this section we focus on classic random unweighted graphs, namely,
we will consider ER, WS (small world model), and BA graphs (scale-free model), and calculate probability $p$ that the random graph is shift enabled. We examine how $p$ depends on the parameters used to generate the graph.

\subsection{Generic Random Graphs Models}
\begin{figure}[htb]
    \centering
    \includegraphics[width=1\linewidth]{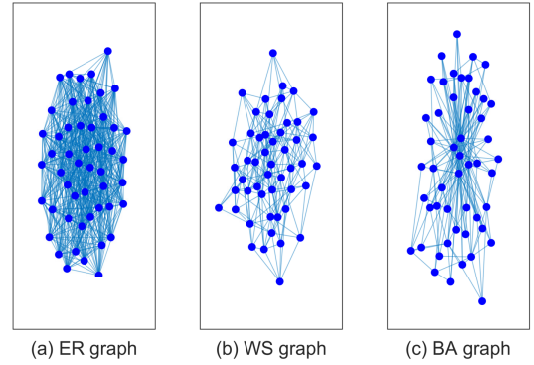}
    \caption{Examples of generic random graph models with $N$=50 nodes. (a) ER graph with the probability of generating edges $P=0.5$. 
    ER random graph has no central node, and most nodes are evenly connected.
    (b) WS graph with average degree $K=4$ and rewriting probability $\beta=0.5$. WS graph has small-world properties, including short average path lengths and high clustering. (c) BA graph with the initial number of nodes $m_0=5$ and the number of edges for each node addition $m=4$.
    BA graph is a scale-free network, in which the connections between nodes are severely unevenly distributed: a few nodes in the network called Hub points have extremely many connections, while most nodes have only a small number of connections.
    }
    \label{fig:ERWSBA}
\end{figure}
\subsubsection{Erdős–Rényi random graph (Figure~\ref{fig:ERWSBA} (a))}
The Erdős–Rényi (ER) random graph is often used to model many real-world inference problems, and it has been used in the context of graph filter design, e.g., in \cite{segarra2017network}, \cite{mei2016signal} and \cite{nassif2018distributed}. The ER graph topology can be defined in two ways~\cite{frieze2015introduction} as:
\begin{itemize}
	\item 
	$\mathcal{G}(N,P)$ where $N$ and $P$ denote  the number of nodes and the probability that an edge is present, respectively.
	\item
	$\mathcal{G}(N,M)$ where $M$ is the number of edges, which are randomly distributed in the graph.
\end{itemize}
The relationship between the two models
is $P=M/\binom{N}{2}$, 
and the two models are asymptotic equivalent as $N$ increase (Theorem 1.4 in Ref.~\cite{frieze2015introduction}).	
\subsubsection{Watts-Strogatz random graph (Figure~\ref{fig:ERWSBA} (b))}
The Watts-Strogatz (WS) model is a classic random graph generation model which produces graphs with small-world properties \cite{Watts1998Collective}.
WS graph has three parameters, $N$, $K$ and $\beta$, and is denoted as $\mathcal{G}=WS(N, K,\beta)$ where $N$ is the number of nodes, $K$ is the average degree of nodes, and $\beta$ is the rewriting probability. If $\beta=0$ WS graph is a regular ring lattice in which each node is connected to the nearest $K$ nodes, and if $\beta=1$, the WS graph becomes an ER graph.

\subsubsection{Barabási–Albert random graph (Figure~\ref{fig:ERWSBA} (c))}
The Barabási–Albert (BA) model is an algorithm for generating random scale-free networks using a preferential attachment mechanism. It can model many practical scale-free networks 
such as the World Wide Web, citation networks and social networks. In contrast to the ER and the WS models, the BA model is scale-free, and the connections between the nodes are severely unevenly distributed: a few nodes in the network, called Hub points, have extremely many connections, while most other nodes have only a small number of connections. The
BA random model has three parameters: the number of nodes $N$, the initial number of nodes $m_0$, and the number of edges for each node addition $m$, $m\leq m_0$, and is denoted by $\mathcal{G}=BA(N,m_0, m)$. New nodes are added according to the priority strategy $-$ ``the more connections between nodes, the greater the possibility of receiving a new link" ~\cite{barabasi1999emergence,albert2002statistical,barabasi2013network}.

\subsection{
Shift-enabled properties of ER graphs}\label{sec: ER graph}
 
Figure~\ref{fig:the shift-enabled probability of ER graphs} shows the probability $p$ that an ER graph $\mathcal{G}(N,M)$ is shift-enabled as a function of the number of edges $M$. The simulation results are provided for $N$=50 nodes and the results are averaged over $10^5$ runs. From the figure, we can identify three distinct regions: Region 1: a very small $M$ where $p$ is zero; Region 2: the flat region when $p$ reached the maximum close to 1; Region 3, where $p$ drops to 0 for very large $M$, following a symmetric trend to Region 1.

Next, based on these heuristic findings, we separately treat the three regions. We fix the number of nodes $N$ and change the number of edges $M$, and theoretically analyse the behaviour of the probability $p$ that the resulting graph is shift-enabled. 

\begin{figure}[htb]
	\centering
	\includegraphics[width=2.5in]{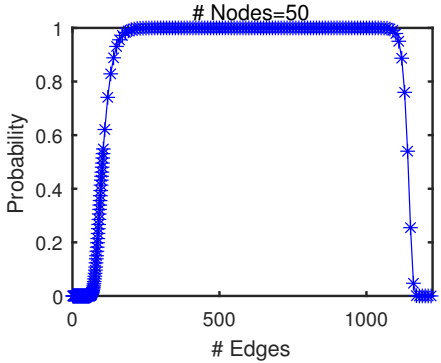}
	\caption{The shift-enabled probability of the unweighted ER graph with $N$=50 nodes: probability that the graph is shift-enabled $p$ vs. the number of edges $M$.}        \label{fig:the shift-enabled probability of ER graphs}
\end{figure}


\subsubsection{Region 1 $-$ $M$ is small}\label{sec: Region 1}

In this region, based on the following theorem, the probability of a graph being shift-enabled is always zero.

\begin{Thm}{\label{Thm_p_tends_to_zero}}
	If a graph $\mathcal{G}$ is unconnected then the shift matrix is not shift-enabled.
\end{Thm}
\begin{proof}
	Assume $\mathcal{G}$ has two unconnected components ${\mathcal{G}}_1$ and ${\mathcal{G}}_2$, with corresponding Laplacian matrices $L_1$ and $L_2$, respectively. Since both $L_1$ and $L_2$ have an eigenvalue equal to zero, $L$, the Laplacian matrix of $\mathcal{G}$ will have the eigenvalue 0 with at least the multiplicity of two. Therefore, based on Lemma~\ref{Lem_undirected_shift}, $\mathcal{G}$ is non-shift-enabled.
\end{proof}
Since an $N$-node connected graph has at least  $N-1$ edges, we have the following corollary.
\begin{Cor} {\label{Cor：non shift for small edges}}
	If $M\leq N-2$ then the shift matrix $L$ of graph $\mathcal{G}$ is non-shift-enabled. 
\end{Cor}
According to Corollary~\ref{Cor：non shift for small edges} the probability of a graph being shift-enabled is $p=0$ when $M$ is small with respect to $N$, that is, when $M\leq N-2$.



\subsubsection{Region 2 $-$ Moderate $M$}\label{sec: Region 2}


The following theorem shows that when $N$ is sufficiently large, and the probability that an edge is present $P$ is far from 0 and 1 (alternatively, $M$ is moderately large) then the eigenvalues of Laplacian matrix are distinct. The uniqueness of the eigenvalues guarantees the shift-enable property by Lemma~\ref{Lem_undirected_shift}.
\begin{Thm}[Distinct eigenvalue, Theorem 1.3 in Ref. \cite{tao2017random}]{\label{Thm: simple spectrum}}
	Let $\mathcal{G}$ be a connected graph and $X_N=(x_{i,j})_{1\leq i,j \leq N}$ be an $N\times N$ real symmetric random matrix in which the upper-triangular entries $x_{i,j} (i<j)$ are independent ((see Remark~\ref{remark:upper triangular} (a))
	and have non-trivial distribution  for some fixed $\mu>0$ (see Remark ~\ref{remark:upper triangular} (b)). 
	Furthermore, 
	$x_{i,i}$ are independent of the upper diagonal entries $x_{i,j} (1\leq i<j\leq N)$ (see details in Remark~\ref{remark:upper triangular} (c)).
	Then for every fixed constant $c>0$ and $N$ sufficiently
large (depending on $c$ and $\mu$), the eigenvalues of $X_N$ are distinct with probability at
least $1-N^{-c}$. 
That is, for sufficiently large $N$, the probability that the eigenvalues of $X_N$ are distinct tends to 1.
\end{Thm}

\begin{figure}[htb]
	\centering
	\subfigure[]{
	\includegraphics[width=0.48\linewidth]{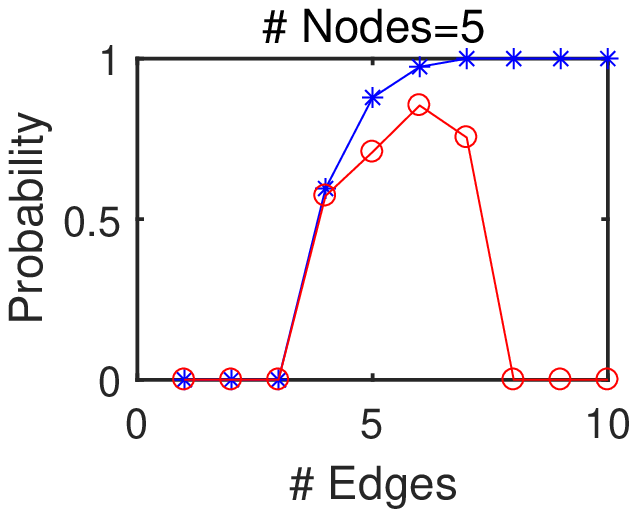}}
	\subfigure[]{
	\includegraphics[width=0.48\linewidth]{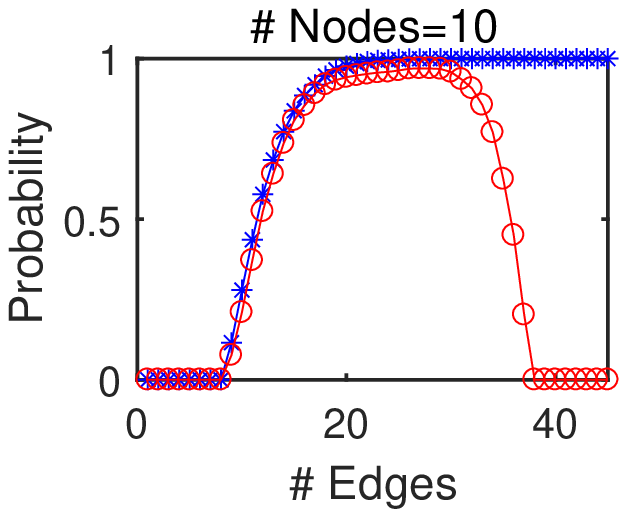}}
	\subfigure[]{
	\includegraphics[width=0.48\linewidth]{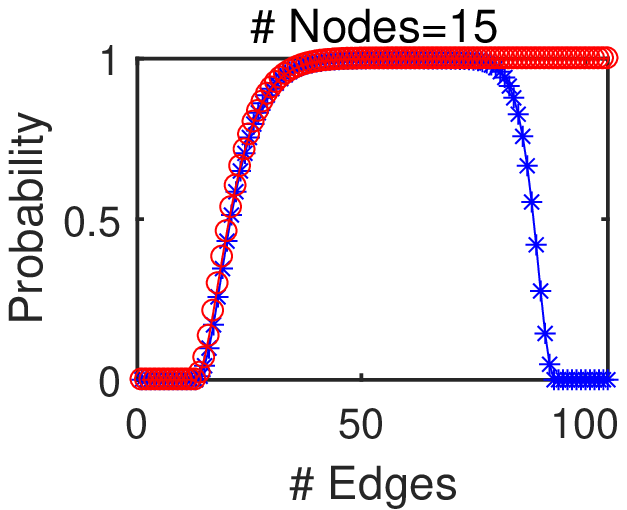}}
	\caption{For an ER graph, the relationship between the probability that the graph is connected (red line) and the shift-enabled probability $p$ (blue line). The horizontal axis shows the number of the edges $M$ and the vertical axis is the probability $p$. (a) $N$=5 nodes; (b) $N$=10 nodes; (c) $N$=15 nodes.}
	\label{fig:connectshiftproblapnode5}
\end{figure}

\begin{Rem}\label{remark:upper triangular}
	(a) The upper-triangular entries are independent since we focus on randomly generated graphs. 
	(b) Non-trivial distribution\footnote{A real-valued random variable $\xi$ is non-trivial if there is a fixed $\mu > 0$
such that $Pr\{\xi=x\}\leq 1-\mu$ (see Equation (1) in Ref.~\cite{tao2017random}).} - elements in Laplacian matrix of an ER graph are non-trivial if $P$ (the probability that an edge is present) stays bounded away from both 0 and 1. If $P=0$ or $P=1$, the entries in the graph Laplacian matrix have trivial distribution, and Theorem~\ref{Thm: simple spectrum} is not applicable.
(c) Since $x_{i,i}=0$ for $1\leq i \leq N$, diagonal entries $x_{i,i}$ are independent of the upper diagonal entries.
\end{Rem}
As a result, the eigenvalues of Laplacian matrix are distinct (the graph is shift-enabled, hence $p=1$) when the number of nodes $N$ is large enough, and the number of edges $M$ is moderately large (very large $M$ implies $P$ close to zero, where Theorem~\ref{Thm: simple spectrum} does not hold).

Our experiments show that the behaviour of $p$ is very similar to the probability of a random graph being connected (see Figure  {\ref{fig:connectshiftproblapnode5}}).
Combining Theorem~\ref{Thm_p_tends_to_zero}, Theorem~\ref{Thm: simple spectrum} and the relation between shift-enabled and connected graphs, we claim that $p$ is close to 1 in Region 2, when the number of nodes $N$ is large and the probability of an edge is bounded away from 0 and 1. When $N$ is very small, Theorem~\ref{Thm: simple spectrum} does not hold; indeed, it can be seen from Figure  {\ref{fig:connectshiftproblapnode5}} that $p$ does not reach 1 for $N$=5 and 10. 

\subsubsection{Region 3 $-$ Very large $M$}\label{sec: Region 3}



We analyse Region 3, that is, the case of very large $M$, by looking graph complement and showing that a graph ${\mathcal{G}}$ and its complement ${\mathcal{G^C}}$ have the same shift-enabled probability. Let $L_{\mathcal{G}}$ and $L_{\mathcal{G^C}}$ denote the Laplacian matrix of ${\mathcal{G}}$ and its complement ${\mathcal{G}^C}$, respectively. 
The following theorem gives the condition that  $\mathcal{G}$ and its complement have the same shift-enabled property.
\begin{Thm}{\label{Thm_complement_shift}}
	If for an $N$-node random graph $\mathcal{G}$, $N$ is not an eigenvalue of $L$ then $\mathcal{G}$ is a shift-enabled graph if and only if $\mathcal{G}^C$ is a shift enabled graph.
\end{Thm}
\begin{proof}
By Lemma~\ref{Lem_complement_connect}, if $N$ is not an eigenvalue of $L$ then $\mathcal{G}^C$ is a connected graph.
	To prove the sufficiency, let $Spec(L_{\mathcal{G}})=(\lambda_1, \lambda_2, \cdots, \lambda_{N-1},\lambda_N=0)$. 
	Then, $Spec(L_{\mathcal{G^C}})=(N-\lambda_1, N-\lambda_2,  \cdots,N-\lambda_{N-1} , 0)$ by  Lemma~\ref{Lem_spec_complement}.
	
	If $\mathcal{G}$ is a shift-enabled graph, then $\lambda_i \neq \lambda_j$  
	according to Lemma~\ref{Lem_undirected_shift}.
	If in addition $N$ is not the eigenvalue of $L$, it can be readily concluded that
	$N-\lambda_i \neq N-\lambda_j$ and $N-\lambda_k \neq 0$  ($1\leq i<j\leq N-1$, $1\leq k\leq N-1$), i.e., all eigenvalues in $L_{\mathcal{G}^C}$ are distinct. Consequently,  $\mathcal{G}^C$ is a shift-enabled graph
	by Lemma~\ref{Lem_undirected_shift}.
	
	The necessity can be easily proven in the  same way.
	To sum up, $\mathcal{G}$ is a shift-enabled graph if and only if $\mathcal{G}^C$ is a shift-enabled graph.
\end{proof}

\begin{Cor}
\label{Cor:R3}
If $N$ is not an eigenvalue of $L$, then, for very large number of edges relative to the number of nodes $N$ ($M>N(N-1)/2-N+2$) the graph is non-shift-enabled.
\end{Cor}

Based on Theorem~\ref{Thm_complement_shift}, the behaviour of $p$ for $\mathcal{G}$ and $\mathcal{G}^C$ is symmetric. Since based on Corollary~\ref{Cor：non shift for small edges}, $\mathcal{G}^C$ is non-shift-enabled when its number of edges is smaller or equal to $N-2$, then, in Region 3, for a very large number of edges, $\mathcal{G}$ is also non-shift-enabled. Note that this implies that fully connected ER graphs are non-shift-enabled. Indeed, as we know, the eigenvalues of a fully connected graph are ${\{-1\}}^{N-1}$ (the multiplicity of $-1$ is $N-1$) and $N-1$, which ensures that the graph is not shift enabled.



The conclusions can be extended to other commonly used shift matrices.
\begin{Rem}
		Commonly used shift matrices of $\mathcal{G}$ are: adjacency matrix $A$, 
		Laplacian matrix $L$, normalized Laplacian matrix $\mathcal{L}$, signless Laplacian matrix $|L|$ and probability transition matrix $T$ which have similar shift-enabled property rules (see Figure~\ref{fig:different shift matrix} for an illustration).
	\begin{figure}[htb]
		\centering
		\includegraphics[width=0.7\linewidth]{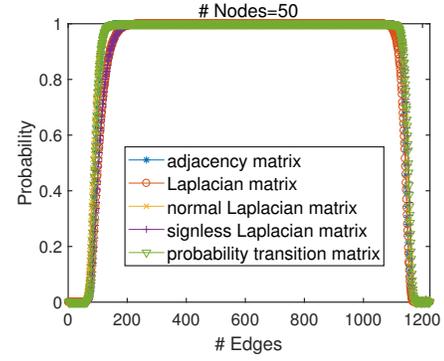}
		\caption{$p$ vs. the number of edges $M$ in an ER random graph with $N$=50 nodes. Adjacency matrix $A$, Laplacian matrix $L$, normal Laplacian matrix $\mathcal{L}$,  signless Laplacian matrix $|L|$ and probability transition matrix $T$ have similar shift-enabled properties.}
		 	\label{fig:different shift matrix}
	\end{figure}
\end{Rem}
\subsection{Shift-enabled properties of WS graphs}
In the previous subsection we discussed how, for fixed $N$, $p$ changes as the number of edges $M$ varies in ER graphs. For an $\mathcal{G}=WS(N, K,\beta)$ graph, in this subsection, we investigate how $p$ depends on WS parameters, namely, $K$, the average degree of nodes and $\beta$ the rewriting probability. Note that the number of edges is $M=N\dot K$. We separately discuss two cases. 
\subsubsection{The number of nodes $N$ fixed, and the average degree $K$ and rewriting probability $\beta$ vary}

\begin{figure}[htb]
	\centering
	\includegraphics[width=0.7\linewidth]{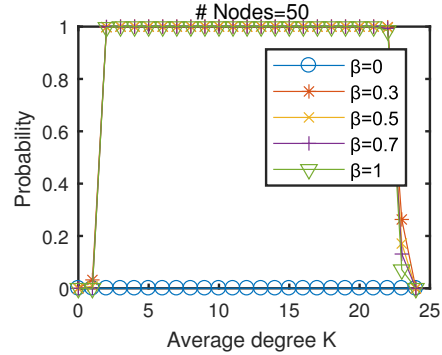}
	\caption{The effect of $\beta$ on the shift-enabled probability $p$ of WS graphs. Except for $\beta=0$ where WS graph is a regular ring lattice, the similar conclusions as for ER graphs can be taken. 
	}
	\label{fig:nodes20changebetalap}
\end{figure}

As shown in Figure~\ref{fig:nodes20changebetalap}, except for $\beta=0$ where the WS graph is a regular ring lattice (see Theorem~\ref{Lem:ring lattice}), the discussions  related to the ER graphs in Section~\ref{sec: Region 1} and Section~\ref{sec: Region 2}
apply to WS graphs as well.
\begin{Thm}{\label{Lem:ring lattice}}
Regular ring lattice is non-shift-enabled when the shift matrix is Laplacian matrix.  
\end{Thm}
\begin{proof}
Regular ring lattice is a special circulant graph in which each node is connected to the nearest $K$ nodes.
Let $A_{cir}$ denote an $N\times N$
circulant matrix whose first row is $[0,1,0,\cdots, 0]$ and the adjacency  matrix of regular ring lattice is $A_{ring}$ whose first row is $[0,\underbrace { {1,1, \cdots ,1} }_K,0,\cdots,0,\underbrace { {1,1, \cdots ,1} }_K]$. Then,
\begin{equation*}
    A_{ring}={\sum_{j=1}^{K}A_{cir}^{j}}+{\sum_{j={N-K}}^{N-1}A_{cir}^{j}}.
\end{equation*}

Since the eigenvalues of $A_{ring}$ are $1,\omega,\omega^{N-1}$, where $\omega=exp(2\pi i/N)=\cos\frac{2\pi }{N}+i\sin\frac{2\pi }{N},$ 
it can be easily concluded that the eigenvalues of $A_{ring}$ are
\begin{flalign*}
 \lambda_j&=\omega^j+
 \cdots+\omega^{K j}+\omega^{(N-K)j}+\cdots+\omega^{(N-1)j}\\   
 &=2[\cos\frac{2\pi j}{N}+\cos\frac{4\pi j }{N}+\cdots+\cos\frac{2K\pi j} {N}],
\end{flalign*}
 where $j=0,1,\cdots,N-1$~\cite{toddh}.
It follows that 
\begin{flalign*}
        &~~\lambda_{N-1}\\
        &=2[\cos\frac{2(N-1)\pi }{N}+\cos\frac{4(N-1)\pi  }{N}+\cdots+\cos\frac{2K(N-1)\pi } {N}]\\
      &=2[\cos\frac{2\pi }{N}+\cos\frac{4\pi  }{N}+\cdots+\cos\frac{2K\pi } {N}]\\
      &=\lambda_1.
\end{flalign*}
Generally, 
$\lambda_0=2K$ and $\lambda_j =\lambda_{N-j}$, for $j=1,2,\cdots,\lfloor \frac{N-1}{2} \rfloor$~\footnote{$\lfloor \frac{N-1}{2} \rfloor$ denotes rounding down $\frac{N-1}{2}$.}.

 Furthermore, the Laplacian matrix of the regular ring lattice is $L_{ring}=D_{ring}-A_{ring}$ and the degree matrix $D_{ring}=2K\times I_N$ \footnote{$I_N$ is an identity matrix of size $N \times N$.}. Therefore, the eigenvalues of $L_{ring}$ are $2K-\lambda_j$ for $i=0,1,\cdots,N$ and $2K-\lambda_j =2K-\lambda_{N-j}$ for $j=1,2,\cdots,\lfloor \frac{N-1}{2} \rfloor$.
Then, according to Lemma~\ref{Lem_undirected_shift}, the regular ring lattice is not shift enabled.


\end{proof}
\subsubsection{Rewriting probability $\beta$ fixed and the number of nodes $N$ and the average node degree $K$ vary}

As shown in Figure~\ref{fig:Watt_change_node}, similarly to ER graphs, as the number of nodes increases, the shift-enabled property tends to be stable and has the properties discussed in Section~\ref{sec: Region 1} and Section~\ref{sec: Region 2}. 
\begin{Rem}
Note that, the symmetry in Section~\ref{sec: Region 3} is not as obvious in Figure~\ref{fig:nodes20changebetalap} and  Figure~\ref{fig:Watt_change_node} as it is in Figure~\ref{fig:the shift-enabled probability of ER graphs}. This is because the horizontal axis in WS graph is the average $K$ which is a non-negative even number, while the horizontal axis in Figure~\ref{fig:nodes20changebetalap}
is the number of edges $M$. The symmetry of $p$ of $\mathcal{G}$ and $\mathcal{G}^C$ in Theorem~\ref{Thm_complement_shift} is discussed based on $M$. $K$ and $M$ are linearly dependent via $K=M/N$.
\end{Rem}

\begin{figure}[htb]
	\centering
	\subfigure[]{
		\label{fig:WattFixBeta10node}
		
		\includegraphics[width=1.5in]{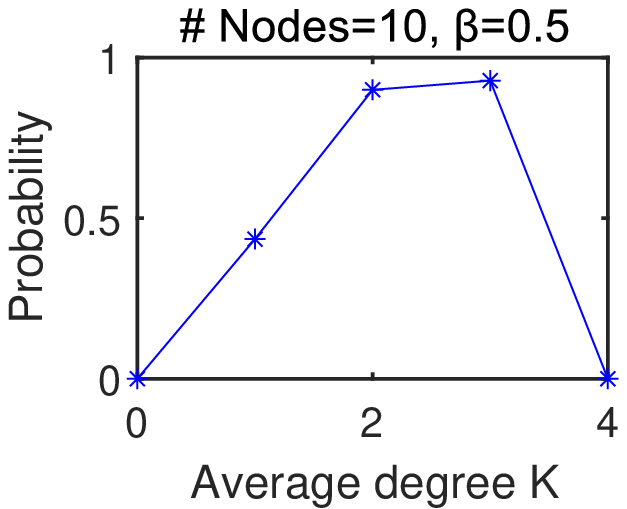}}
	\centering
	\subfigure[]{
		\label{fig:WattFixBeta20node} 
		\includegraphics[width=1.5in]{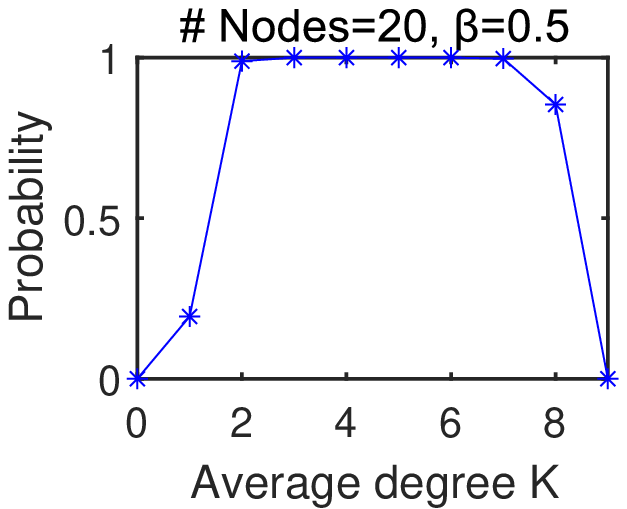}}
	\centering
	\subfigure[]{
		\label{fig:Watt_fix_beta_50node} 
		\includegraphics[width=1.5in]{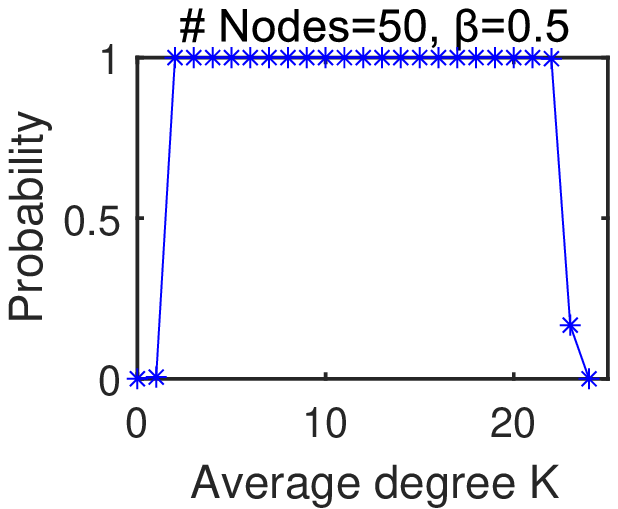}}
	\caption{The shift-enabled probability $p$ of WS graph as a function of $K$. With the increase of the number of nodes, $p$ follows the same trends as for ER graphs. (a) $N=$10 nodes. 
		(b) $N=$20 nodes. (c) $N=$50 nodes. 
	}
	\label{fig:Watt_change_node}
\end{figure}

\subsection{Shift-enabled properties of BA graphs}
\begin{figure}[htb]
	\centering
	\subfigure[]{
		\includegraphics[width=1.5in]{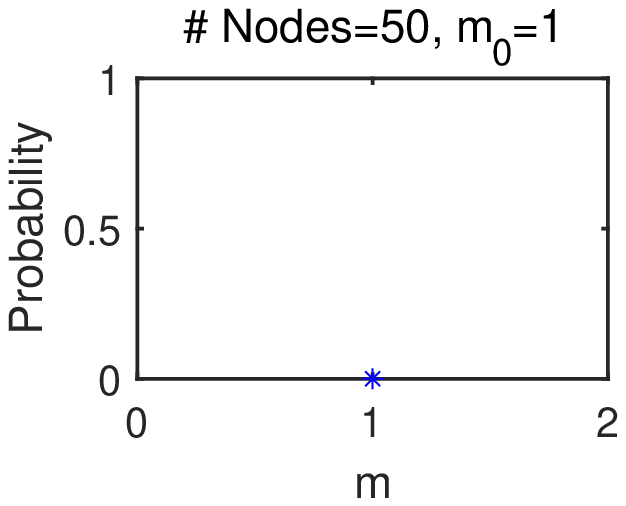}}
	\quad

	\centering
	\subfigure[]{
		\includegraphics[width=1.5in]{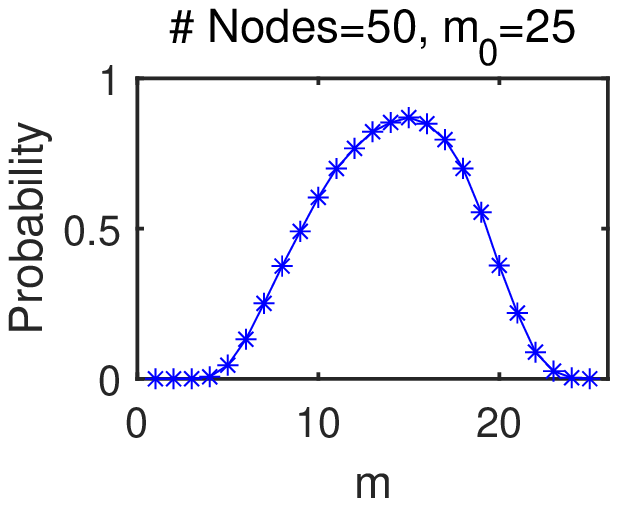}
	}

	\caption{The effect of $m_0$ on the shift-enabled probability $p$ in a BA graph with $N=50$.
		(a) $m_0=1$; 
		(b) $m_0=25$;
	}
	\label{fig：BA_fix_change_m0}
\end{figure}
	\begin{figure}[htb]
	    \centering
	    \includegraphics[width=2in]{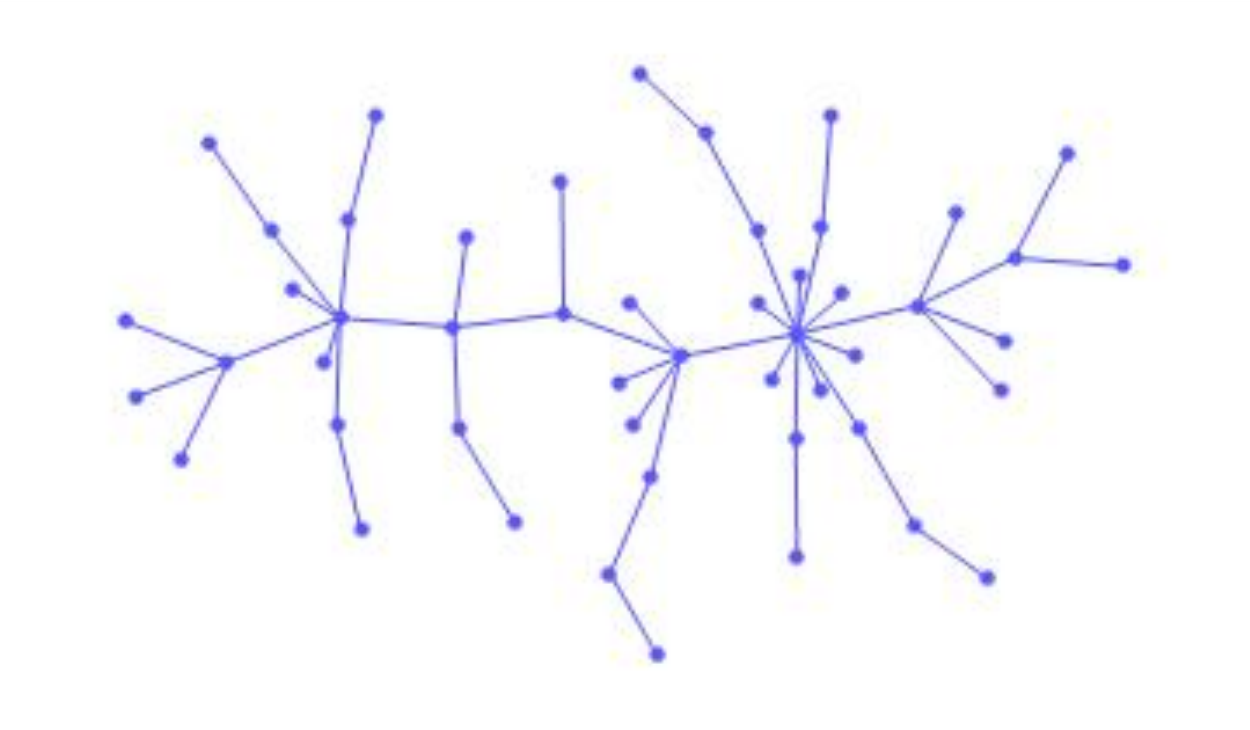}
	    \caption{The BA graph with $m_0=1$ and $N=50$. As $m=1$ the graph is a tree graph, in which case the  main components of the random model contain only a small number of nodes.} 
	    \label{fig:BA graph with m0_1}
	\end{figure}
\begin{figure}[htb]
    \centering
    \includegraphics[width=2.5in]{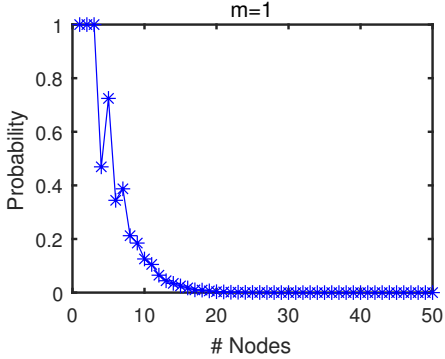}
    \caption{BA graph: $p$ vs. the number of nodes $N$ for $m_0=m=1$.}
    \label{fig:BAChangeN_m0_1}
\end{figure}
In Figure~\ref{fig：BA_fix_change_m0}, we fixed the total number of nodes $N$, and change the initial number of nodes $m_0$ from 1 to $N$ to study the probability $p$ of a graph being shift-enabled for a range of $m_0$ values.
 The BA graph used in the experiment was generated by $gsp\_barabasi\_albert.m$ function available in GSPBox~\cite{perraudin2014gspbox}. 
 Note that the total number of edges is $M=(N-m_0)\times m$.  Since $m\leq m_0$, the maximum total number of edges is $|M|_{max}=(N-m_0)\times m_0$. For fixed $N$, $|M|_{max}$ is $m_0=\frac{N}{2}$. We provide the results for $m_0=\frac{N}{2}=25$.
 
 Based on Figure~\ref{fig：BA_fix_change_m0}, we make the following conclusions:
 \begin{itemize}
     \item  As can be seen from Figure~\ref{fig：BA_fix_change_m0} (a), when $m_0=1$, the number of edges for each node addition, $m$, can only be equal to 1 (since $m\leq m_0$). In this case the graph is a tree graph (Figure~\ref{fig:BA graph with m0_1}),  in which few nodes have many connections, while most nodes have very few connections, which reflects the scale-free feature of the BA graph. As shown in Figure~\ref{fig:BAChangeN_m0_1}, as the number of nodes $N$ increases, the connectivity of the vertex is increasingly concentrated, and the probability $p$ tends to 0.


 \item When $m_0=25$ and $m=1$, $M=25<N-1=49$, the graph is unconnected and hence $p=0$. As $m$ increases, the probability that the graph is connected increases, and as a result, $p$ increases.
\end{itemize}

\section{Extension to Weighted and Signed graphs}
\label{sec:weights}
In this section we extend the results to random weighted and signed graphs. 

\subsection{Weighted graphs}
Many practical scenarios are modelled by weighted graphs, where Gaussian distribution or Exponential distribution~\cite{spiegel2001probability} are often used to define the edge weights. Figures~\ref{fig:graph compare with with exponentially distributed weights} and \ref{fig:graph compare with Gaussian distributed weights} show how $p$ varies with the number of edges $M$ for the three considered random graph models with $N$=50 nodes.


 


\begin{figure}[htb]
	\centering
	\subfigure[]{
    		\includegraphics[width=1.5in]{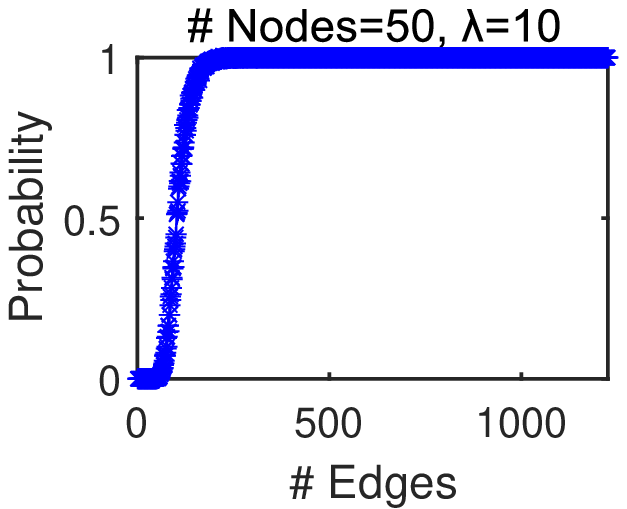}}
	\subfigure[]{
		\includegraphics[width=1.5in]{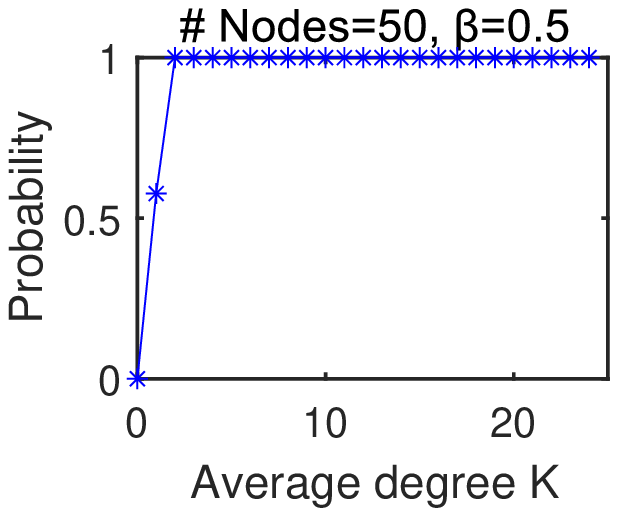}}
	\subfigure[]{
		\includegraphics[width=1.5in]{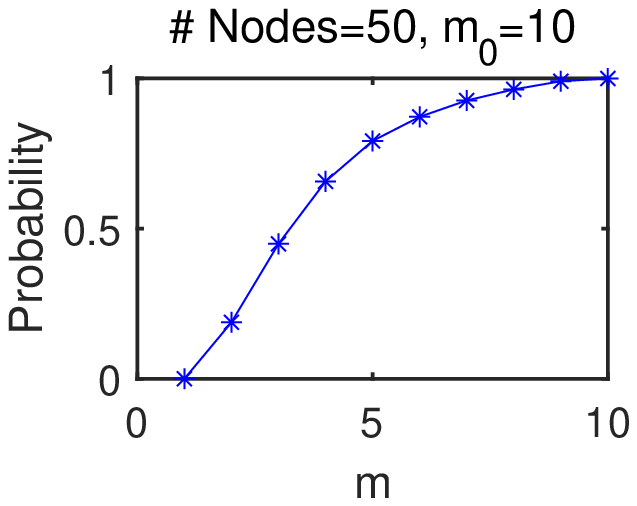}}
	\caption{The shift-enabled probability $p$ of weighted graphs with exponentially distributed weights with $\lambda=10$. (a){ ER graph; (b) WS graph;} (C) BA graph.}
	\label{fig:graph compare with with exponentially distributed weights}
\end{figure}

\begin{figure}[htb]
	\centering
	\subfigure[]{
		\includegraphics[width=1.5in]{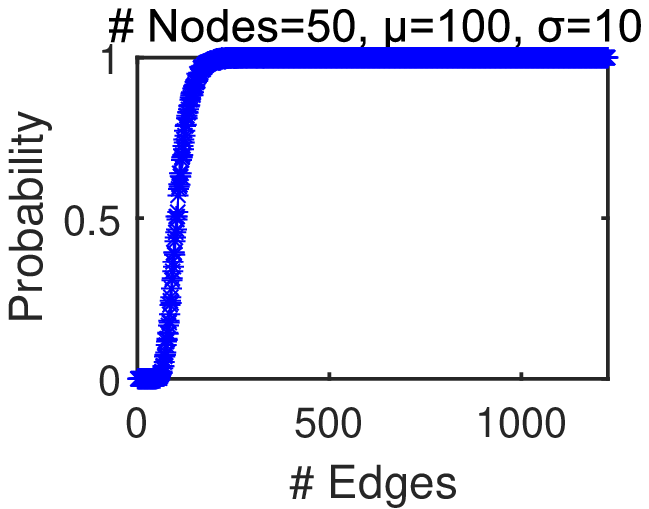}}
	\subfigure[]{
		\includegraphics[width=1.5in]{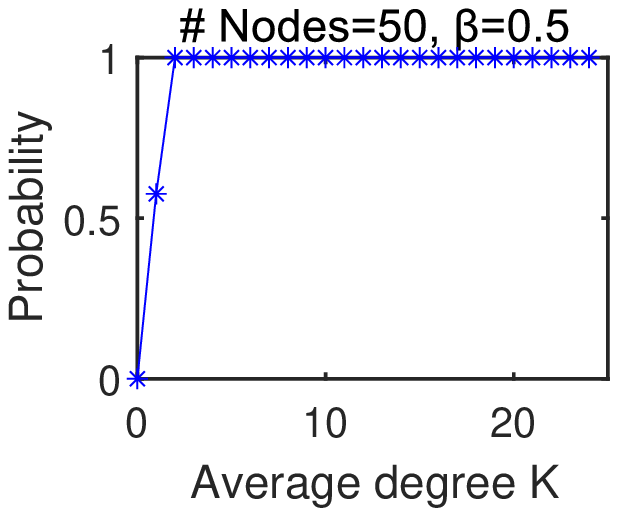}}
	\subfigure[]{
		\includegraphics[width=1.5in]{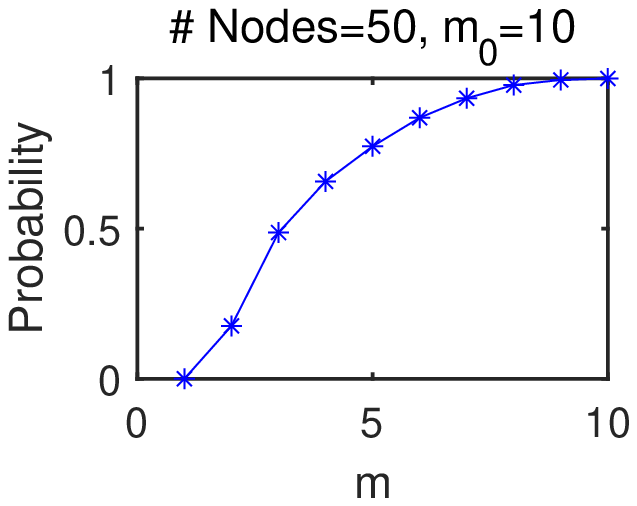}}
	\caption{The shift-enabled probability $p$ of weighted graphs with Gaussian distributed weights with $\mu=100$ and $\sigma=10$. {(a) ER graph; (b) WS graph;} (C) BA graph.}
	\label{fig:graph compare with Gaussian distributed weights}
\end{figure}

\begin{Thm}{\label{Thm:general law_weighted_graph}}
Theorem \ref{Thm_p_tends_to_zero}, Corollary \ref{Cor：non shift for small edges} and Theorem~\ref{Thm: simple spectrum} still hold for weighted graphs. Therefore, the probability $p$ of the graph being shift-enabled has the following properties:
\begin{itemize}
	\item {Region 1.} The probability $p$ is $0$, when the number of edges $M$ is relatively small which in accordance with Section~\ref{sec: Region 1} for an unweighted graph.
	\item {Region 2.}
	The probability $p$ is close to 1 when the number of edges increases which is in accordance to Section~\ref{sec: Region 2}.
	\end{itemize}
\end{Thm}

In the last region, $p$ remains close $1$, which is different from the shift-enabled probability behaviour of the unweighted graph. This is due to the following theorem.
\begin{Thm} {\label{Thm: weighted graph with large edges}}
In the weighted graph,  the probability $p$ is close to $1$ and does not drop to zero for a large number of edges $M$.
\end{Thm}
\begin{proof}
The results in Section~\ref{sec: Region 3} do no longer hold for the weighted graph since Theorem~\ref{Thm_complement_shift} does not hold. Furthermore, the complements of weighted graph are generally connected as the original graph is connected (Complement of undirected weighted graph are defined in detail in Section 2.7.1 of Ref.~\cite{networks}), 
which implies that Theorem~\ref{Thm: simple spectrum} still holds for the complement graph. So, the probability $p$ remains close to $1$.
\end{proof}

Exponential distribution function $f(x)=\lambda e^{-\lambda x},(x>0)$ has one parameter $\lambda$. As Figure \ref{fig:erexpnode20mu100lap} shows for $\lambda=0.1, 1, 10, 100$, and the parameter  $\lambda$ has almost no effect on the shift-enabled properties.
\begin{figure}[htb]
	\centering
	\includegraphics[width=0.5\linewidth]{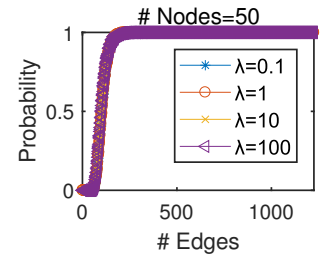}
	\caption{ER graphs with the weights defined using exponential distribution. The influence of parameter $\lambda$ on the shift-enabled probability $p$, for an ER weighted graph with $N$=50 nodes. The horizontal axis shows the number of edges $M$.
	}. 
	\label{fig:erexpnode20mu100lap}
\end{figure}

In Gaussian distribution, the distribution function $ f(x)={\frac {1}{\sigma {\sqrt {2\pi }}}}e^{-{\frac {1}{2}}\left({\frac {x-\mu }{\sigma }}\right)^{2}} $, in which $\mu$ and $\sigma>0$ are mean and standard deviation of random variable. The weight may be negative, and the graphs with negative weights are called sign graphs which is discussed in Section~\ref{section: signed graph}. As Figure \ref{fig:ergaussiannode20sigama10lap} shows, (for positive weights) the parameters of Gaussian distribution have almost no effect on the shift-enabled properties.
\begin{figure}[htb]
	\centering
	\subfigure[]{
	\includegraphics[width=0.48\linewidth]{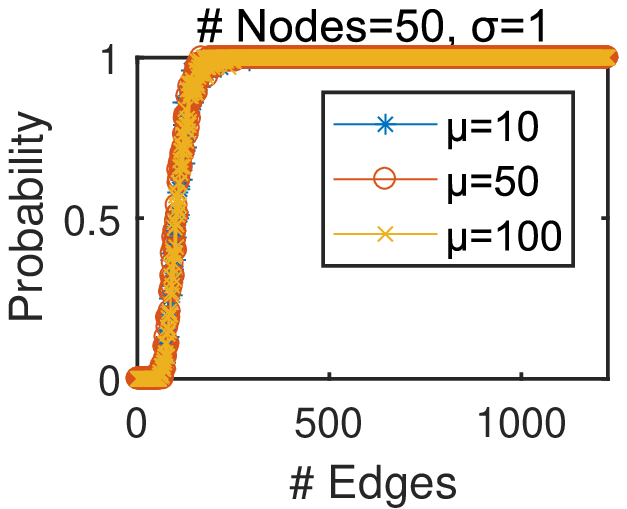}}
	\centering
	\subfigure[]{
	\includegraphics[width=0.48\linewidth]{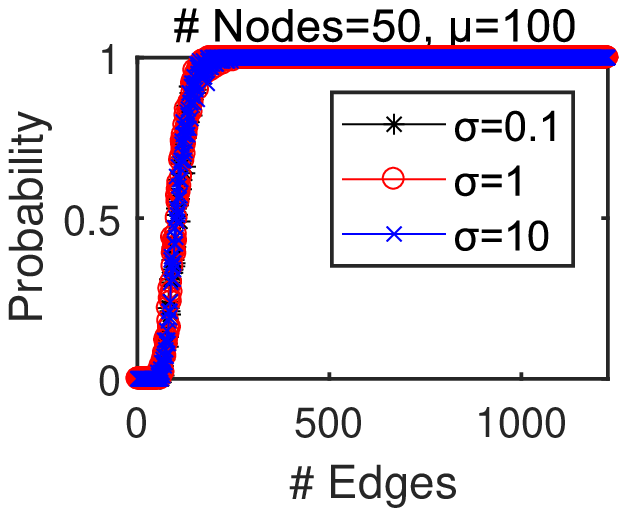}}
	\caption{ER graphs with the weights defined using Gaussian distribution. Probability $p$ vs the number of edges $M$ for three different values of (a) $\mu$ (b) $\sigma$.}
	\label{fig:ergaussiannode20sigama10lap}
\end{figure}

\subsection{Signed Graphs}{\label{section: signed graph}}
Up to now, only graphs with unweighted or positively weighted edges are discussed. Some applications, however, involve signed graphs in which case negative weights are admitted.
In these cases, the negative edges are usually used to measure the dissimilarity between nodes. 

Signed graphs are widely used in social networks, in which case person support/oppose each other, recommendation networks (likes/dislikes), and biological networks (promote and inhibit relationships between neurons)~\cite{cheung2018robust,parisien2008solving,Dittrich20200signedgraph}. This article studies the most simplest signed graphs with weights equal to $1$ or $-1$ introduced by Harary~\cite{harary1953notion} in 1953, to deal with social relations including  disliking, indifference, and liking. The results from the previous section, can easily be extended to signed graphs.

The Laplacian matrix of signed graphs is defined as $L_{sign}=D_{sign}-W$, where $D_{sign}=\sum_j|w_{ij}|$, i.e., the sum of the absolute weights of incident edges~\cite{Dittrich20200signedgraph,liu2020generalized}. 
Signed graphs have two categories: balanced and unbalanced graphs. 
A signed graph is balanced if the product of edge weights around each cycle is positive~\cite{harary1953notion}. Otherwise, the graph is unbalanced. 
Balanced and unbalanced signed graphs have the distinct shift-enabled properties. Next, the two cases are considered separately.

For balanced graphs, Laplacian matrix $L_{sign}$ has the following property.
\begin{Thm}[\cite{kunegis2011spectral}]{\label{Thm balance_spectral}}
	The eigenvalues of the Laplacian matrices of a balanced graph $\mathcal{G}_{sign}$ and the corresponding unsigned graph $\mathcal{G}$ 
	are identical. 
\end{Thm}

Based on this theorem it is easy to show the following shift-enabled properties the signed graphs.

\begin{Thm}
  Balanced graphs and unbalanced graphs 
  have the same shift-enabled properties as the corresponding unsigned graphs and general random weighted graphs, respectively.
\end{Thm}
\begin{proof}
First of all, based on Lemma~\ref{Lem_undirected_shift} and  Theorem~\ref{Thm balance_spectral}, it is obvious that a balanced signed graph has the same shift-enabled properties as the corresponding unsigned graph. In the considered case, the signed graph has weights $1$ and $-1$. So, its shift-enabled properties are in accordance with the corresponding unweighted graph (see Figure~\ref{fig:balance_unbalance}(a) for a simulation example).

For unbalanced graphs, it is easy to show that Theorem~\ref{Thm:general law_weighted_graph} and Theorem~\ref{Thm: weighted graph with large edges} still hold, so its shift-enabled properties are in accordance with those of random weighted graphs (see Figure~\ref{fig:balance_unbalance}(b)).
\end{proof}

\begin{Rem}
Random balanced graphs and unbalanced graphs are generated differently. 
A usual way to generate a random unbalanced graph, is to start from  an ER graph, and then randomly replace all weights that are 1 by -1, i.e., to convert the unsigned ER graph into a signed. Because balanced graphs require that the product of edge weights around each cycle is positive, which is difficult to satisfy when the number of edges is high, this method is not practical for generating balanced graphs. 


 According to Lemma~\ref{Lem:balance_graph_condition}(3), a balanced signed graph can be generated in the following way. Firstly, divide randomly the vertices into two sets $V_1$ and $V_2$; then the weights of edges connecting two vertices that are both in $V_1$ or both in $V_2$ are set to $1$; edges connecting one vertex from $V_1$ with a vertex to $V_2$ are set to $-1$. 
\end{Rem}

\begin{figure}[htb]
	\centering
	\subfigure[]{
		\includegraphics[width=1.5in]{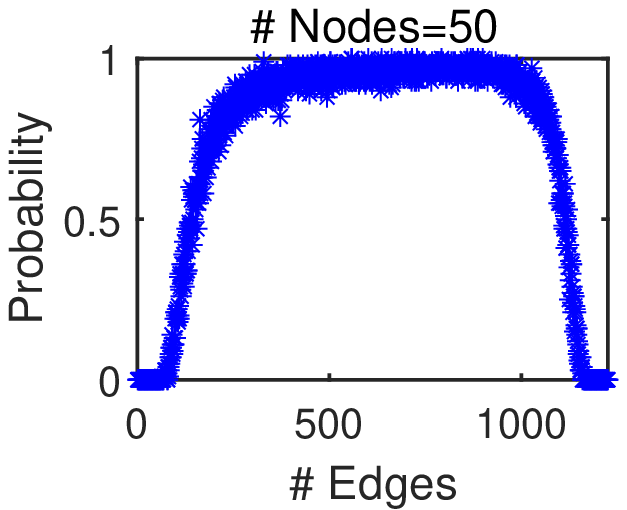}}
	\subfigure[]{
		\includegraphics[width=1.5in]{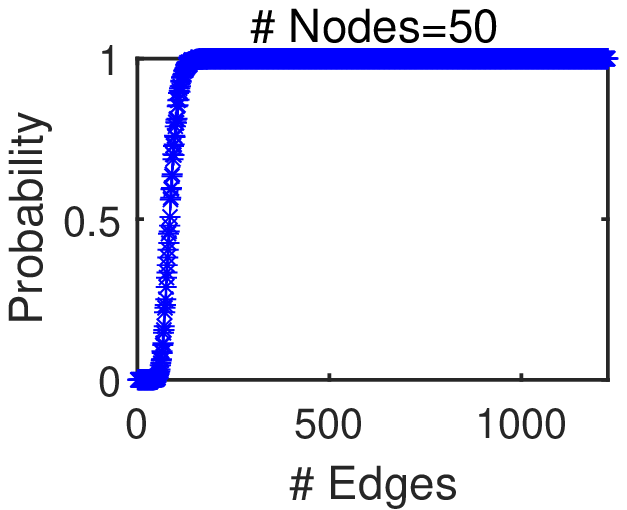}}
	\caption{The shift-enabled properties of signed  graph.
	Balanced graph has the same shift-enabled properties as the corresponding unsigned  graph.
	Unbalanced graph`s shift-enabled properties are in accordance with those of general weighted  graphs. $p$ vs. $M$, with $N$=50 nodes, for a signed
		 (a) Balanced; (b) Unbalanced  graph.}
	\label{fig:balance_unbalance}
\end{figure}



\section{Conclusion}
\label{sec:discussion}
Recognising the importance of shift-enable property, the paper discusses when a random graph is shift-enabled. The behaviour of the shift-enabled property for weighted and unweighted graphs was discussed and the the design guidelines that ensure the shift-enable property were given. 
A future direction is to extend the findings to directed graphs and how to transform shift matrix when the shift-enabled condition for $S$ is not satisfied.
It is also interesting to study if one may decompose non-shift-enabled graphs into shift-enabled subgraphs, to optimize the design of the GSP filters.

\appendices

\section{}
It is easily determined whether a graph is shift-enabled by the following lemma.

\begin{Lem}{\label{Lem_undirected_shift}}
	If shift matrix $L$ is a real symmetric matrix, then $L$ is shift-enabled, if and only if all eigenvalues of $L$ are distinct {\rm\cite{ortega2018graph,Liyanchen2021}}.
\end{Lem}

Lemma~\ref{Lem_undirected_shift} indicates that an undirected graph, which shift matrix is real symmetric, is shift-enabled if and only if its eigenvalues are all distinct.
\begin{Lem}{\label{Lem_spec_complement}}{\rm{\cite{Merris1994Laplacian}}}
	Assume 
	\begin{equation*}
	Spec(L_{\mathcal{G}})=(\lambda_1, \lambda_2, \cdots, \lambda_{N-1},\lambda_N=0),
	\end{equation*}
	then
	\begin{equation*}
	Spec(L_{\mathcal{G^C}})=(N-\lambda_1, N-\lambda_2,  \cdots,N-\lambda_{N-1}, 0).
	\end{equation*} 
\end{Lem}
\begin{proof}
Since $\lambda_i \neq 0$ for $i=1,\cdots, N-1$,
$u=[1,1,\cdots,1]^T$ is the only eigenvector of $L_{\mathcal{G}}$ with eigenvalue $0$. Let $x_1,x_2,\cdots,x_{N-1}$ be the remaining eigenvectors of $L_{\mathcal{G}}$. As $L_{\mathcal{G}}$ is symmetric and $\lambda_N \neq \lambda_i$ for $i\neq N$, $u$ is orthogonal to $x_1,x_2,\cdots,x_{N-1}$.

	Let $L_{K_N}$ be the Laplacian matrix of a complete graph with $N$ vertices. It is easily seen that
	\begin{equation}{\label{Equ_Laplacain_complement}}
	L_{\mathcal{G}}+L_{\mathcal{G}^C}=L_{K_n}=NI_{N}-J_{N}, 
	\end{equation}
	where $I_N$ and $J_{N}$ are identity matrix and all-1s matrix with $N$ vertices, respectively. 

	
Therefore, according to Equation (\ref{Equ_Laplacain_complement}), 
	\begin{equation*}
	L_{\mathcal{G}^C}x_i=(NI_N-J_N-L_{\mathcal{G}})x_i=Nx_i-J_Nx_i-{\lambda}_i x_i=(N-{\lambda}_i)x_i,
	\end{equation*} 
	where $J_N x_i = 0$ as $x_i$ is orthogonal to $u$. Thus,  
	$(N-{\lambda}_i) \in Spec(L_{\mathcal{G}^C})$ for $i=1,\cdots, N-1$.
	
		Finally, $0 \in Spec(L_{\mathcal{G^C}})$ as $u=[1,1,\cdots,1]^T$ is still an eigenvector of $L_{\mathcal{G}^C}$ with eigenvalue $0$. 
Hence,
	$Spec(L_{\mathcal{G^C}})=(N-\lambda_1, N-\lambda_2,  \cdots,N-\lambda_{N-1} , 0)$.
\end{proof}
\begin{Lem}{\label{Lem_complement_connect}}{\rm{\cite{Merris1994Laplacian}}}
	If $N$ is not the eigenvalue of $L_{\mathcal{G}}$, 
	then $\mathcal{G}^C$ is a connected graph.
\end{Lem}

\begin{Lem}\label{Lem:balance_graph_condition}
	%
	Let $\mathcal{G}$ be a connected signed graph, i.e., a graph whose non-zero edge weights can take positive and negative values. 
	The following conditions are equivalent~\cite{zaslavsky1982signed,hou2003laplacian,Dittrich20200signedgraph}:
	\begin{itemize}
		\item [(1)] $\mathcal{G}$ is balanced. 
		\item [(2)] $L$ is positive definite, i.e., all the eigenvalues of $L$ are positive.
		\item [(3)] The vertices of $\mathcal{G}$ can be divided into two subsets: $V_1$ and $V_2$, such that all edges connecting vertices that are both in $V_1$ and all edges connecting the vertices that are both in $V_2$ are $+1$, whereas the edges connecting a vertex in $V_1$ with a vertex in $V_2$ are $-1$.
	\end{itemize}
\end{Lem}

\section*{Acknowledgment}
The authors thank Bochao Zhao and Kanghang He for helpful discussions.
This research was supported by the European Union's Horizon 2020 research and innovation programme under the Marie Sklodowska-Curie grant agreement  number 734331, the National Key Research and Development Project under grant 2017YFE0119300 and general scientific research projects of Zhejiang Education Department in 2020 under grant Y202044676.

\ifCLASSOPTIONcaptionsoff
  \newpage
\fi




\bibliographystyle{IEEEtran}
\bibliography{reference}

\begin{thebibliography}{10}
\providecommand{\url}[1]{#1}
\csname url@samestyle\endcsname
\providecommand{\newblock}{\relax}
\providecommand{\bibinfo}[2]{#2}
\providecommand{\BIBentrySTDinterwordspacing}{\spaceskip=0pt\relax}
\providecommand{\BIBentryALTinterwordstretchfactor}{4}
\providecommand{\BIBentryALTinterwordspacing}{\spaceskip=\fontdimen2\font plus
\BIBentryALTinterwordstretchfactor\fontdimen3\font minus
  \fontdimen4\font\relax}
\providecommand{\BIBforeignlanguage}[2]{{%
\expandafter\ifx\csname l@#1\endcsname\relax
\typeout{** WARNING: IEEEtran.bst: No hyphenation pattern has been}%
\typeout{** loaded for the language `#1'. Using the pattern for}%
\typeout{** the default language instead.}%
\else
\language=\csname l@#1\endcsname
\fi
#2}}
\providecommand{\BIBdecl}{\relax}
\BIBdecl

\bibitem{Liyanchen2018}
L.~Chen, S.~Cheng, V.~Stankovic, and L.~Stankovic, ``Shift-enabled graphs:
  Graphs where shift-invariant filters are representable as polynomials of
  shift operations,'' \emph{IEEE Signal Processing Letters}, vol.~25, no.~9,
  pp. 1305--1309, 2018.

\bibitem{sandryhaila_2014_big_data}
A.~Sandryhaila and J.~M. Moura, ``Big data analysis with signal processing on
  graphs: Representation and processing of massive data sets with irregular
  structure,'' \emph{IEEE Signal Processing Magazine}, vol.~31, no.~5, pp.
  80--90, 2014.

\bibitem{Shuman_2013_The_emerging_field}
D.~I. Shuman, S.~K. Narang, P.~Frossard, A.~Ortega, and P.~Vandergheynst, ``The
  emerging field of signal processing on graphs: Extending high-dimensional
  data analysis to networks and other irregular domains,'' \emph{IEEE Signal
  Proc. Magazine}, vol.~30, no.~3, pp. 83--98, 2013.

\bibitem{ortega2018graph}
A.~Ortega, P.~Frossard, J.~Kova{\v{c}}evi{\'c}, J.~M. Moura, and
  P.~Vandergheynst, ``Graph signal processing: Overview, challenges, and
  applications,'' \emph{Proceedings of the IEEE}, vol. 106, no.~5, pp.
  808--828, 2018.

\bibitem{segarra2017optimal}
S.~Segarra, A.~G. Marques, and A.~Ribeiro, ``Optimal graph-filter design and
  applications to distributed linear network operators,'' \emph{IEEE
  Transactions on Signal Processing}, vol.~65, no.~15, pp. 4117--4131, 2017.

\bibitem{coutino2019advances}
M.~Coutino, E.~Isufi, and G.~Leus, ``Advances in distributed graph filtering,''
  \emph{IEEE Transactions on Signal Processing}, vol.~67, no.~9, pp.
  2320--2333, 2019.

\bibitem{sandryhaila_2013_discrete}
A.~Sandryhaila and J.~M. Moura, ``Discrete signal processing on graphs,''
  \emph{IEEE Transactions on Signal Processing}, vol.~61, no.~7, pp.
  1644--1656, 2013.

\bibitem{segarra2017network}
S.~Segarra, A.~G. Marques, G.~Mateos, and A.~Ribeiro, ``Network topology
  inference from spectral templates,'' \emph{IEEE Transactions on Signal and
  Information Processing over Networks}, vol.~3, no.~3, pp. 467--483, 2017.

\bibitem{Liyanchen2021}
L.~Chen, S.~Cheng, K.~He, V.~Stankovic, and L.~Stankovic, ``Undirected graphs:
  Is the shift-enabled condition trivial or necessary?'' \emph{IEEE Access},
  pp. 1--1, 2021.

\bibitem{sandryhaila_2014_discrete_frequency}
A.~Sandryhaila and J.~Moura, ``Discrete signal processing on graphs: Frequency
  analysis,'' \emph{IEEE Transactions on Signal Processing}, vol.~62, no.~12,
  pp. 3042--3054, 2014.

\bibitem{poignard2018spectra}
C.~Poignard, T.~Pereira, and J.~P. Pade, ``Spectra of laplacian matrices of
  weighted graphs: structural genericity properties,'' \emph{SIAM Journal on
  Applied Mathematics}, vol.~78, no.~1, pp. 372--394, 2018.

\bibitem{lancaster_1985_matrix_theory}
P.~Lancaster and M.~Tismenetsky, \emph{The Theory of Matrices: with
  Applications}.\hskip 1em plus 0.5em minus 0.4em\relax Elsevier, 1985.

\bibitem{Marques_2017}
A.~Marque, S.~Segarra, G.~Leus, and A.~Ribeiro, ``Stationary graph processes
  and spectral estimation,'' \emph{IEEE Transactions on Signal Processing},
  vol.~65, pp. 5911--5926, 2017.

\bibitem{dong2020graph}
X.~Dong, D.~Thanou, L.~Toni, M.~Bronstein, and P.~Frossard, ``Graph signal
  processing for machine learning: A review and new perspectives,'' \emph{IEEE
  Signal Processing Magazine}, vol.~37, no.~6, pp. 117--127, 2020.

\bibitem{Dittrich20200signedgraph}
T.~{Dittrich} and G.~{Matz}, ``Signal processing on signed graphs\:
  Fundamentals and potentials,'' \emph{IEEE Signal Processing Magazine},
  vol.~37, no.~6, pp. 86--98, 2020.

\bibitem{mei2016signal}
J.~Mei and J.~M. Moura, ``Signal processing on graphs: Causal modeling of
  unstructured data,'' \emph{IEEE Transactions on Signal Processing}, vol.~65,
  no.~8, pp. 2077--2092, 2016.

\bibitem{nassif2018distributed}
R.~Nassif, C.~Richard, J.~Chen, and A.~H. Sayed, ``Distributed diffusion
  adaptation over graph signals,'' in \emph{2018 IEEE International Conference
  on Acoustics, Speech and Signal Processing (ICASSP)}.\hskip 1em plus 0.5em
  minus 0.4em\relax IEEE, 2018, pp. 4129--4133.

\bibitem{frieze2015introduction}
A.~Frieze and M.~Karo{\'n}ski, \emph{Introduction to random graphs}.\hskip 1em
  plus 0.5em minus 0.4em\relax Cambridge University Press, 2015.

\bibitem{Watts1998Collective}
D.~J. Watts and S.~H. Strogatz, ``Collective dynamics of 'small-world'
  networks.'' \emph{Nature}, 1998.

\bibitem{barabasi1999emergence}
A.-L. Barab{\'a}si and R.~Albert, ``Emergence of scaling in random networks,''
  \emph{science}, vol. 286, no. 5439, pp. 509--512, 1999.

\bibitem{albert2002statistical}
R.~Albert and A.-L. Barab{\'a}si, ``Statistical mechanics of complex
  networks,'' \emph{Reviews of modern physics}, vol.~74, no.~1, p.~47, 2002.

\bibitem{barabasi2013network}
A.-L. Barab{\'a}si, ``Network science,'' \emph{Philosophical Transactions of
  the Royal Society A: Mathematical, Physical and Engineering Sciences}, vol.
  371, no. 1987, p. 20120375, 2013.

\bibitem{tao2017random}
T.~Tao and V.~Vu, ``Random matrices have simple spectrum,''
  \emph{Combinatorica}, vol.~37, no.~3, pp. 539--553, 2017.

\bibitem{toddh}
J.~TODD and C.~WALL, ``H. bass, jfc kingman, f. smithies.''

\bibitem{perraudin2014gspbox}
N.~Perraudin, J.~Paratte, D.~Shuman, L.~Martin, V.~Kalofolias,
  P.~Vandergheynst, and D.~K. Hammond, ``Gspbox: A toolbox for signal
  processing on graphs,'' \emph{arXiv preprint arXiv:1408.5781}, 2014.

\bibitem{spiegel2001probability}
M.~R. Spiegel, J.~J. Schiller, R.~A. Srinivasan, and M.~LeVan,
  \emph{Probability and statistics}.\hskip 1em plus 0.5em minus 0.4em\relax
  Mcgraw-hill, 2001, vol.~2.

\bibitem{networks}
\BIBentryALTinterwordspacing
B.~Hadorn. {Graph Theory}. (2016, September 11). [Online]. Available:
  \url{http://www.xatlantis.ch/index.php/education/zeus-framework/15-graph-theory}
\BIBentrySTDinterwordspacing

\bibitem{cheung2018robust}
G.~Cheung, W.-T. Su, Y.~Mao, and C.-W. Lin, ``Robust semisupervised graph
  classifier learning with negative edge weights,'' \emph{IEEE Transactions on
  Signal and Information Processing over Networks}, vol.~4, no.~4, pp.
  712--726, 2018.

\bibitem{parisien2008solving}
C.~Parisien, C.~H. Anderson, and C.~Eliasmith, ``Solving the problem of
  negative synaptic weights in cortical models,'' \emph{Neural computation},
  vol.~20, no.~6, pp. 1473--1494, 2008.

\bibitem{harary1953notion}
F.~Harary \emph{et~al.}, ``On the notion of balance of a signed graph.''
  \emph{Michigan Mathematical Journal}, vol.~2, no.~2, pp. 143--146, 1953.

\bibitem{liu2020generalized}
J.-B. Liu, J.~Zhao, and Z.-Q. Cai, ``On the generalized adjacency, laplacian
  and signless laplacian spectra of the weighted edge corona networks,''
  \emph{Physica A: Statistical Mechanics and Its Applications}, vol. 540, p.
  123073, 2020.

\bibitem{kunegis2011spectral}
J.~Kunegis, ``On the spectral evolution of large networks,'' 2011.

\bibitem{Merris1994Laplacian}
R.~Merris, ``Laplacian matrices of graphs: a survey,'' \emph{Linear Algebra Its
  Applications}, vol. 197-198, no.~2, pp. 143--176, 1994.

\bibitem{zaslavsky1982signed}
T.~Zaslavsky, ``Signed graphs,'' \emph{Discrete Applied Mathematics}, vol.~4,
  no.~1, pp. 47--74, 1982.

\bibitem{hou2003laplacian}
Y.~Hou, J.~Li, and Y.~Pan, ``On the laplacian eigenvalues of signed graphs,''
  \emph{Linear and Multilinear Algebra}, vol.~51, no.~1, pp. 21--30, 2003.

\end{thebibliography}

\end{document}